\documentclass[aps,prl,twocolumn,superscriptaddress]{revtex4}
\usepackage{bm}
\usepackage{graphicx}
\usepackage{amssymb,amsmath,amsbsy,amsgen,amsfonts}     
\usepackage{dcolumn}
\usepackage{amsthm}
\usepackage{mathrsfs}
\usepackage{latexsym}
\usepackage{array}
\usepackage{amstext}
\usepackage{color}
\usepackage{txfonts}

\usepackage{epstopdf} 

\newcommand{\braket}[2]{\langle#1|#2\rangle}

\newcommand{\bra}[1]{\left\langle{#1}\right\vert}
\newcommand{\ket}[1]{\left\vert{#1}\right\rangle}
\newcommand{\ketbra}[2]{|#1\rangle \langle#2|}

\newcommand{\be}{\begin{equation}}
\newcommand{\ee}{\end{equation}}
\newcommand{\ba}{\begin{array}}
\newcommand{\ea}{\end{array}}
\newcommand{\bqa}{\begin{eqnarray}}
\newcommand{\eqa}{\end{eqnarray}}
\newcommand{\ie}{{\it i.e. }}
\newcommand{\tr}{\mbox{Tr}}

\newcommand{\rhalf}{\mbox{$\textstyle \frac{1}{\sqrt{2}}$}}
\newcommand{\norm}[1]{\left|\left|#1\right|\right|}

\newtheorem{theorem}{Theorem}
\newtheorem{corollary}{Corollary}

\newcommand{\COMMENT}[1]{}
\newcommand{\anote}[1]{\textcolor{black}{#1}}
\newcommand{\kb}[1]{\ketbra{#1}{#1}}

\begin{document}

\title{Experimental Verification of Multipartite Entanglement in Quantum Networks}

\author{W. McCutcheon}
\address{Quantum Engineering Technology Laboratory, Department of Electrical and Electronic Engineering, University of Bristol, Woodland Road, Bristol, BS8 1UB, UK}

\author{A. Pappa}
\email{annapappa@gmail.com}
\address{School of Informatics, University of Edinburgh, Edinburgh EH89AB, UK}

\author{B. A. Bell}
\address{Quantum Engineering Technology Laboratory, Department of Electrical and Electronic Engineering, University of Bristol, Woodland Road, Bristol, BS8 1UB, UK}

\author{A. McMillan}
\address{Quantum Engineering Technology Laboratory, Department of Electrical and Electronic Engineering, University of Bristol, Woodland Road, Bristol, BS8 1UB, UK}

\author{A. Chailloux}
\address{INRIA, Paris Rocquencourt, SECRET Project Team, Paris, France}

\author{T. Lawson}
\address{LTCI, CNRS, Telecom ParisTech, Universit\'e Paris-Saclay, 75013 Paris, France}

\author{M.~Mafu}
\address{Department of Physics and Astronomy, Botswana International University of Science and Technology, P/Bag 16 Palapye, Botswana}

\author{D. Markham}
\address{LTCI, CNRS, Telecom ParisTech, Universit\'e Paris-Saclay, 75013 Paris, France}

\author{E. Diamanti}
\address{LTCI, CNRS, Telecom ParisTech, Universit\'e Paris-Saclay, 75013 Paris, France}

\author{I. Kerenidis}
\address{CNRS IRIF, Universit\'e Paris 7, Paris 75013, France}
\address{Centre for Quantum Technologies, National University of Singapore, 3 Science Drive 2, Singapore 117543, Singapore}

\author{J. G. Rarity}
\address{Quantum Engineering Technology Laboratory, Department of Electrical and Electronic Engineering, University of Bristol, Woodland Road, Bristol, BS8 1UB, UK}

\author{M. S. Tame} 
\email{markstame@gmail.com}
\address{School of Chemistry and Physics, University of KwaZulu-Natal, Durban 4001, South Africa}
\address{National Institute for Theoretical Physics, University of KwaZulu-Natal, Durban 4001,
South Africa}

\date{\today}

\begin{abstract}
Multipartite entangled states are a fundamental resource for a wide range of quantum information processing tasks. In particular, in quantum networks it is essential for the parties involved to be able to verify if entanglement is present before they carry out a given distributed task. Here we design and experimentally demonstrate a protocol that allows any party in a network to check if a source is distributing a genuinely multipartite entangled state, even in the presence of untrusted parties. The protocol remains secure against dishonest behaviour of the source and other parties, including the use of system imperfections to their advantage. We demonstrate the verification protocol in a three- and four-party setting using polarization-entangled photons, highlighting its potential for realistic photonic quantum communication and networking applications.
\end{abstract}


\maketitle

\section{Introduction}

Entanglement plays a key role in the study and development of quantum information theory and is a vital component in quantum networks~\cite{Horodecki,Kimble,Scarani,Chiri,Pers}. The advantage provided by entangled states can be observed, for example, when the quantum correlations of the $n$-party Greenberger-Horne-Zeilinger (GHZ) state~\cite{GHZ} are used to win a nonlocal game with probability 1, while any classical local theory can win the game with probability at most $3/4$ (see Ref.~[7]). In a more general setting, multipartite entangled states allow the parties in a network to perform distributed tasks that outperform their classical counterparts~\cite{Buh}, to delegate quantum computation to untrusted servers~\cite{broadbent:focs09}, or to compute through the measurement-based quantum computation model~\cite{raussendorf:prl01}. It is therefore vital for parties in a quantum network to be able to verify that a state is entangled, especially in the presence of untrusted parties and by performing only local operations and classical communication.

A protocol for verifying that an untrusted source creates and shares the $n$-qubit multipartite entangled GHZ state, $\ket{GHZ_n}=\frac{1}{\sqrt{2}}\big(|0\rangle^{\otimes n}+|1\rangle^{\otimes n}\big)$, with $n$ parties has recently been proposed~\cite{Pappa12}. In the verification protocol, the goal of the honest parties is to determine how close the state they share is to the ideal GHZ state and verify whether or not it contains genuine multipartite entanglement (GME) -- entanglement that can only exist if all qubits were involved in the creation of the state~\cite{Horodecki}. On the other hand, any number of dishonest parties that may collaborate with the untrusted source are trying to `cheat' by convincing the honest parties that the state they share is close to the ideal GHZ state and contains GME when this may not be the case. Verifying GME in multipartite GHZ states in this way is relevant to a wide variety of protocols in distributed quantum computation and quantum communication. While distributed quantum computation is at an early stage of development experimentally~\cite{Barz11,Barz13,Greganti16}, many schemes for using multipartite GHZ states in distributed quantum communication have already been demonstrated, including quantum secret sharing~\cite{Tittel01}, open-destination teleportation~\cite{Zhao} and multiparty quantum key distribution~\cite{Chen05,Adamson06}. This makes the entanglement verification protocol relevant for distributed quantum communication with present technology.

\begin{figure*}[t]
\includegraphics[width=9.5cm]{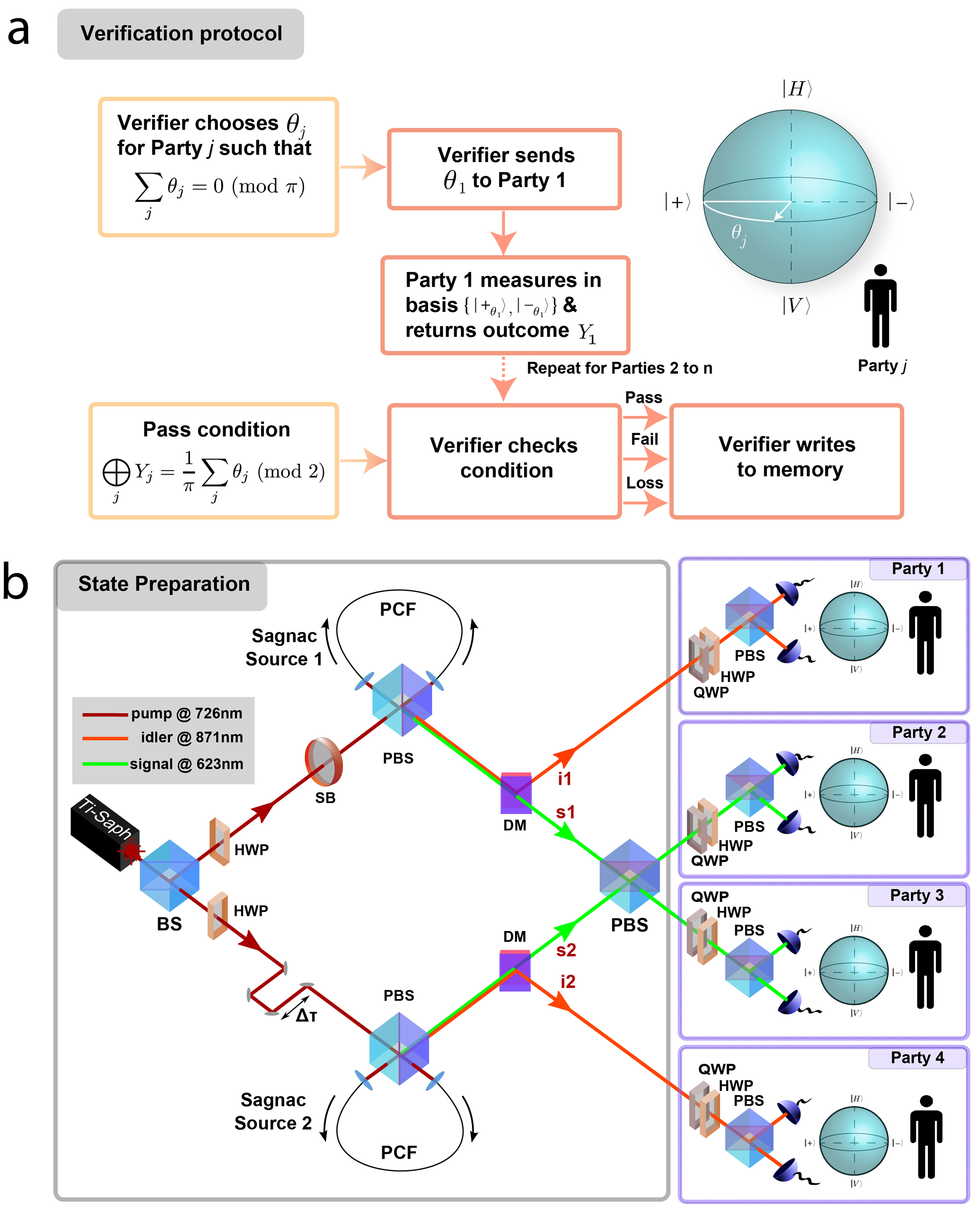}
\caption{The verification protocol and experimental setup. {\bf a}, A flow diagram showing the steps of the verification protocol. {\bf b,} The experimental setup for state preparation, consisting of a femto-second laser (Spectra-Physics Tsunami) filtered to give $1.7$~nm bandwidth pulses at $726$~nm. The laser beam is split by a beamsplitter (BS) into two modes with the polarisation set to diagonal by half-wave plates (HWPs). One mode undergoes a temporal offset, $\Delta T$, using a translation stage and the other a phase rotation using a Soleil-Babinet (SB) compensator. The modes each enter a photonic crystal fibre (PCF) source via a polarising beamsplitter (PBS) in a Sagnac configuration, enabling pumping in both directions. The sources generate non-degenerate entangled signal and idler photon pairs by spontaneous four-wave mixing. Temperature tuning in one of the sources is used to match the spectra of the resulting signal photons in the other source. The entangled photon pairs exit the sources via the PBS and due to their non-degenerate wavelengths they are separated by dichroic mirrors (DMs) and filtered with $\Delta \lambda_s=40$nm at $\lambda_s = 623$~nm (tunable $\Delta \lambda_i = 2$~nm at $\lambda_i= 871$~nm) in the signal (idler) to remove any remaining light from the pump laser. The signal photons from each pair interfere at a PBS and all photons are collected into single-mode fibres. Pairs of automated half- and quarter-wave plates (QWPs) on each of the four output modes from the fibres allow arbitrary rotations to be made before the modes are split by PBSs and the light is detected by eight silicon avalanche photodiode detectors (APDs). The protocol's software (outlined in panel (a)) is linked to an 8-channel coincidence counting box (Qumet MT-30A) and the automated wave plates in order to set each unique measurement basis for the parties and detect single-shot four-fold coincidences.}
\label{fig:ExperimentalPlanPlayers}
\end{figure*} 

In order for a quantum protocol to be practical, however, it must take into account system imperfections, including loss and noise, throughout the protocol (generation, transmission and detection of the quantum state). In previous work~\cite{Pappa12}, it was shown that by using a suitable protocol, the closeness of a shared resource state to a GHZ state and the presence of GME can be verified in a distributed way between untrusted parties under perfect experimental conditions. However, the protocol is not tolerant to arbitrary loss and in fact it cannot be used for a loss rate that exceeds 50\%. 

In this work, we design and experimentally demonstrate a protocol that outperforms the original one in Ref.~[11]. We examine quantitatively how a dishonest party can use system imperfections to boost their chances of cheating and show our protocol defends against such tactics. We demonstrate both the original and new protocols using a source of polarization-entangled photons, which produces three- and four-party GHZ states, and examine the performance of the protocols under realistic experimental conditions. Our results are perfectly adapted to photonic quantum networks and can be used to reliably verify multipartite entanglement in a real-world quantum communication setting. In order to achieve verification of a state in an untrusted setting, the protocols exploit the capability of GHZ states to produce extremal correlations which are unobtainable by any quantum state that is not locally equivalent to the GHZ state. This property has been shown to bound state fidelities in the fully device independent setting of nonlocality via self-testing~\cite{Mayers04,Pal14,McKague11}. In addition, a related recent study~\cite{cavalcanti:natcom15} has proposed a method to detect multipartite entanglement in the `steering' setting in which some of the devices are known to be untrusted (or defective), by using one-sided device-independent entanglement witnesses. Our protocols extend beyond these methods by allowing the \emph{amount} of entanglement to be quantified in terms of an appropriate fidelity measure in a setting where some unknown parties are untrusted, as well as providing a method for dealing with loss and other inefficiencies in the system. This makes our protocols and analysis more appropriate for a realistic network setting.

\section{Results}

\subsection{The verification protocol}

The network scenario we consider consists of a source that shares an $n$-qubit state $\rho$ with $n$ parties, where each party receives a qubit. One of the parties, a `Verifier', would like to verify how close this shared state is to the ideal state and whether or not it contains GME. The protocol to do this is as follows: First, the Verifier generates random angles $\theta_j \in [0,\pi)$ for all parties including themselves ($j\in[n]$), such that $\sum_j \theta_j$ is a multiple of $\pi$. The angles are then sent out to all the parties in the network. When party $j$ receives their angle from the Verifier they measure in the basis $\{\ket{+_{\theta_j}},\ket{-_{\theta_j}}\}=\{\frac{1}{\sqrt{2}}(\ket{0}+e^{i\theta_j}\ket{1}),\frac{1}{\sqrt{2}}(\ket{0}-e^{i\theta_j}\ket{1})\}$ and send the outcome $Y_j=\{0,1\}$ to the Verifier. A flow diagram of the protocol is shown in Fig.~\ref{fig:ExperimentalPlanPlayers}a, where the order in which the angles are sent out and outcomes returned is irrelevant and it is assumed that the Verifier and each of the parties share a secure private channel for the communication. This can be achieved by using either a one-time pad or quantum key distribution~\cite{Scarani}, making the communication secure even in the presence of a quantum computer. The state passes the test when the following condition is satisfied: if the sum of the randomly chosen angles is an even multiple of $\pi$, there must be an even number of $1$ outcomes for $Y_j$, and if the sum is an odd multiple of $\pi$, there must be an odd number of $1$ outcomes for $Y_j$. We can write this condition as
\be
\bigoplus_j Y_j=\frac{1}{\pi}\sum_j\theta_j\pmod 2.
\ee

For an ideal $n$-qubit GHZ state, the test succeeds with probability 1 (see Supplementary Note 1). Moreover, it can be shown that the fidelity $F(\rho)=\bra{GHZ_n}\rho\ket{GHZ_n}$ of a shared state $\rho$ with respect to an ideal GHZ state can be lower bounded by a function of the probability of the state passing the test, $P(\rho)$. If we first suppose that all $n$ parties are honest, then $F(\rho)\geq 2P(\rho)-1$ (see Supplementary Note 1). Furthermore, we can say that GME is present for a state $\rho$ when $F(\rho)>1/2$ with respect to an ideal GHZ state~\cite{TothGuhne} and therefore GME can be verified when the pass probability is $P(\rho)>3/4$. This verification protocol, that we will call the `$\theta$-protocol', is a generalisation of the protocol in Ref.~[11], called the `$XY$-protocol', where the angles $\theta_j$ are fixed as either $0$ or $\pi/2$, corresponding to measurements in the Pauli $X$ or $Y$ basis. In the honest case and under ideal conditions, the lower bound for the fidelity is the same in both protocols.

When the Verifier runs the test in the presence of $n-k$ dishonest parties, the dishonest parties can always collaborate and apply a local or joint operation $U$ to their part of the state. This encompasses the different ways in which the dishonest parties may try to cheat in the most general setting. Hence, we look at a fidelity measure given by $F'(\rho)=\max_U F\big((\mathbb{I}_k\otimes U_{n-k})\rho(\mathbb{I}_k\otimes U^\dag_{n-k})\big)$, and lower bound it by the pass probability as $F'(\rho)\geq 4P(\rho)-3$ for both the $\theta$ and $XY$ protocols (see Supplementary Note 1). This gives directly a bound of $P(\rho)>7/8=0.875$ to observe GME. However, by concentrating on attacks for the case $F'(\rho)=1/2$, tighter analysis can be performed (see Supplementary Note 1), where the GME bound can be shown to be $P(\rho) \geq 1/2 +1/\pi \approx 0.818$ for the $\theta$-protocol and $P(\rho) \geq \cos^2(\pi/8)\approx 0.854$ for the $XY$ protocol. The $\theta$-protocol is more sensitive to detecting cheating and hence can be used to verify GME more broadly in realistic implementations where the resources are not ideal.

The above bounds do not account for loss. To analyse cheating strategies which take advantage of loss we must allow the dishonest parties (which have potentially perfect control of the source and their equipment) to choose to declare `loss' at any point. In particular they may do this when they are asked to make measurements that would reduce the probability of success, making the round invalid, which can skew the statistics in favour of passing to the advantage of the dishonest parties. This may change the fidelity and GME bounds above. We address this to find GME bounds in the case of loss in our photonic realization.

\begin{figure*}[t]
\includegraphics[width=15cm]{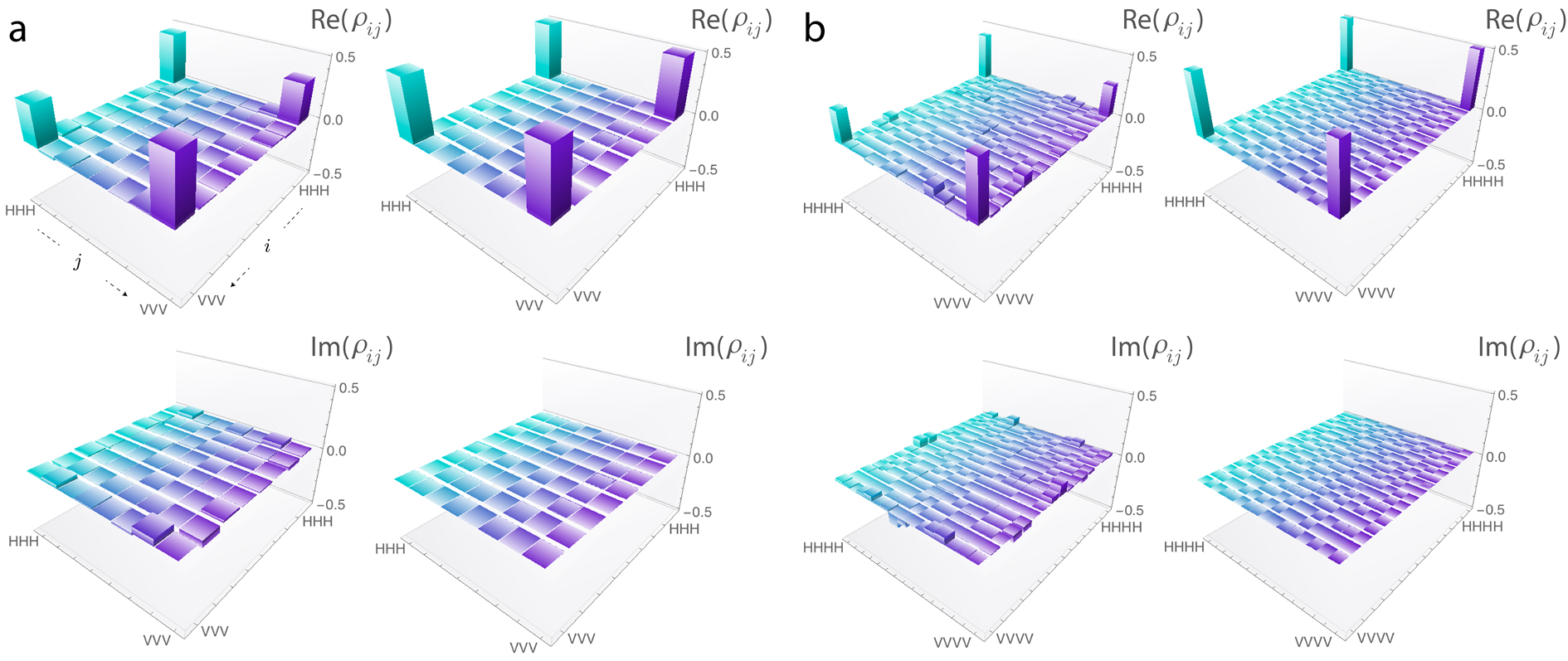}
\caption{Tomographic reconstruction of the three- and four-photon GHZ states used in the protocols. {\bf a,} Three-photon GHZ state (left column) and ideal case (right column). {\bf b,} Four-photon GHZ state (left column) and ideal case (right column). Top row corresponds to the real parts and bottom row corresponds to the imaginary parts. The density matrix elements are given by $\rho_{ij}=\bra{i}\rho_{exp} \ket{j}$, where $\rho_{exp}$ is the reconstructed experimental density matrix.}
\label{fig:DensMats4}
\end{figure*} 

\subsection{Experimental setup}

The optical setup used to perform the verification protocols is shown in Fig.~\ref{fig:ExperimentalPlanPlayers}b. The source of GHZ states consists of two micro-structured photonic crystal fibers (PCFs), each of which produces a photon pair by spontaneous four-wave mixing, with the signal wavelength at $623$~nm and the idler at $871$~nm  (see Supplementary Note 2). To generate entangled pairs of photons, each fibre loop is placed in a Sagnac configuration, where it is pumped in both directions. When the pump pulse entering the Sagnac loop is in diagonal polarisation, conditional on a single pair being generated by the pump laser the state exiting the polarising beamsplitter (PBS) of the loop is in the Bell state $\frac{1}{\sqrt{2}}\left(\ket{H}_{\rm s}\ket{H}_{\rm i}+\ket{V}_{\rm s}\ket{V}_{\rm i}\right)$, with ${\rm s}$ and ${\rm i}$ indicating the signal and idler photons, respectively~\cite{Halder09,Clark11}. The signal and idler photons of each source are then separated into individual spatial modes by dichroic mirrors, after which the two signal photons are overlapped at a PBS that performs a parity check, or `fusion' operation~\cite{Bell12,Bell13}. We postselect with $50\%$ probability the detection outcomes in which one signal photon emerges from each output mode of the PBS which projects the state onto the four-photon GHZ state
\be
\frac{1}{\sqrt{2}}\left(\ket{H}_{\rm i_1}\ket{H}_{\rm s_1}\ket{H}_{\rm s_2}\ket{H}_{\rm i_2}+\ket{V}_{\rm i_1}\ket{V}_{\rm s_1}\ket{V}_{\rm s_2}\ket{V}_{\rm i_2}\right). 
\ee
All four photons are then coupled into single-mode fibres, which take them to measurement stages representing the parties in the network. With appropriate angle choices of the wave plates included in these stages any projective measurement can be made by the parties on the polarisation state of their photon~\cite{James01}. In our experiment the successful generation of the state is conditional on the detection of four photons in separate modes, {\it i.e.} postselected. In principle it is possible to move beyond postselection in our setup, where the GHZ states are generated deterministically. This can be achieved by the addition of a quantum non-demolition (QND) measurement of the photon number in the modes after the fusion operation. While technically challenging, QND measurements are possible for photons, for instance as theoretically shown~\cite{Imoto,Xiao} and experimentally demonstrated~\cite{Guerlin}. By using postselection we are able to give a proof-of-principle demonstration of the protocols and gain important information about their performance in such a scenario, including the impact of loss.

In our experiments we use both a three- and a four-photon GHZ state. The generation of the three-photon state requires only a slight modification to the setup, with one of the PCFs pumped in just one direction to generate unentangled pairs (see Supplementary Note 2). Before carrying out the verification protocols we first characterise our experimental GHZ states by performing quantum state tomography~\cite{James01}. The resulting density matrices for the three- and four-photon GHZ states are shown in Fig.~\ref{fig:DensMats4} and have corresponding fidelities $F_{GHZ_3}=0.80 \pm 0.01$ and $F_{GHZ_4}=0.70 \pm 0.01$ with respect to the ideal states. These fidelities compare well with other recent experiments using photons (see Table 1) and are limited mainly by dephasing from the fusion operation~\cite{Bell12} and higher-order emission (see Supplementary Note 2). The errors have been calculated using maximum likelihood estimation and a Monte Carlo method with Poissonian noise on the count statistics, which is the dominant source of error in our photonic experiment~\cite{James01}.

\begin{table}[b]
\centering
\begin{center}
\begin{tabular}{ |c|c|c| } 
\hline
 \hline
 3-photon GHZ Fidelity & 4-photon GHZ Fidelity \\
 \hline 
 $F=0.80 \pm 0.01$, this work & $F=0.70 \pm 0.01$, this work \\ 
 $F=0.768 \pm 0.015$~\cite{Resch05} & $F=0.840 \pm 0.007$~\cite{Zhao03} \\
 $F=0.74 \pm 0.01$~\cite{Zhou08} & $F=0.66 \pm 0.01$~\cite{Bell13} \\ 
 $F=0.811\pm 0.002$~\cite{Lu08} & $F=0.833 \pm 0.004$~\cite{Wang16} \\ 
  $F=0.93\pm 0.01$~\cite{Patel16} & \\ 
  \hline
  \hline
\end{tabular}
\end{center}
\caption{Comparison of GHZ fidelities. The table shows the fidelity of recent three-photon and four-photon GHZ states from other experiments, and includes the fidelities from this work (top row).}
\label{tab:GHZfid} 
\end{table}

\subsection{Entanglement verification}

To demonstrate the verification of multipartite entanglement we use the polarisation degree of freedom of the photons generated in our optical setup. The computational basis states sent out to the parties are therefore defined as $\ket{0}=\ket{H}$ and $\ket{1}=\ket{V}$ for a given photon. Furthermore, the verification protocol relies on a randomly selected set of angles being distributed by the Verifier for each state being tested. To ensure dishonest parties have no prior knowledge, the set of angles is changed after every detection of a copy of the state, {\it i.e.} we perform {\it single-shot} measurements in our experiment. To achieve this, we use automated wave-plate rotators to change the measurement basis defined by the randomised angles for each state. The rotators are controlled by a computer with access to the incoming coincidence data. This approach is needed to provide a faithful demonstration of the protocol and is technologically more advanced than the usual method used in photonic quantum information experiments, where many detections are accumulated over a fixed integration time for a given measurement basis and properties then inferred from the ensemble of states. We now analyse the performance of the $XY$ and $\theta$ verification protocols for the three- and four-party GHZ states.

\begin{figure*}[t]
\includegraphics[width=13cm]{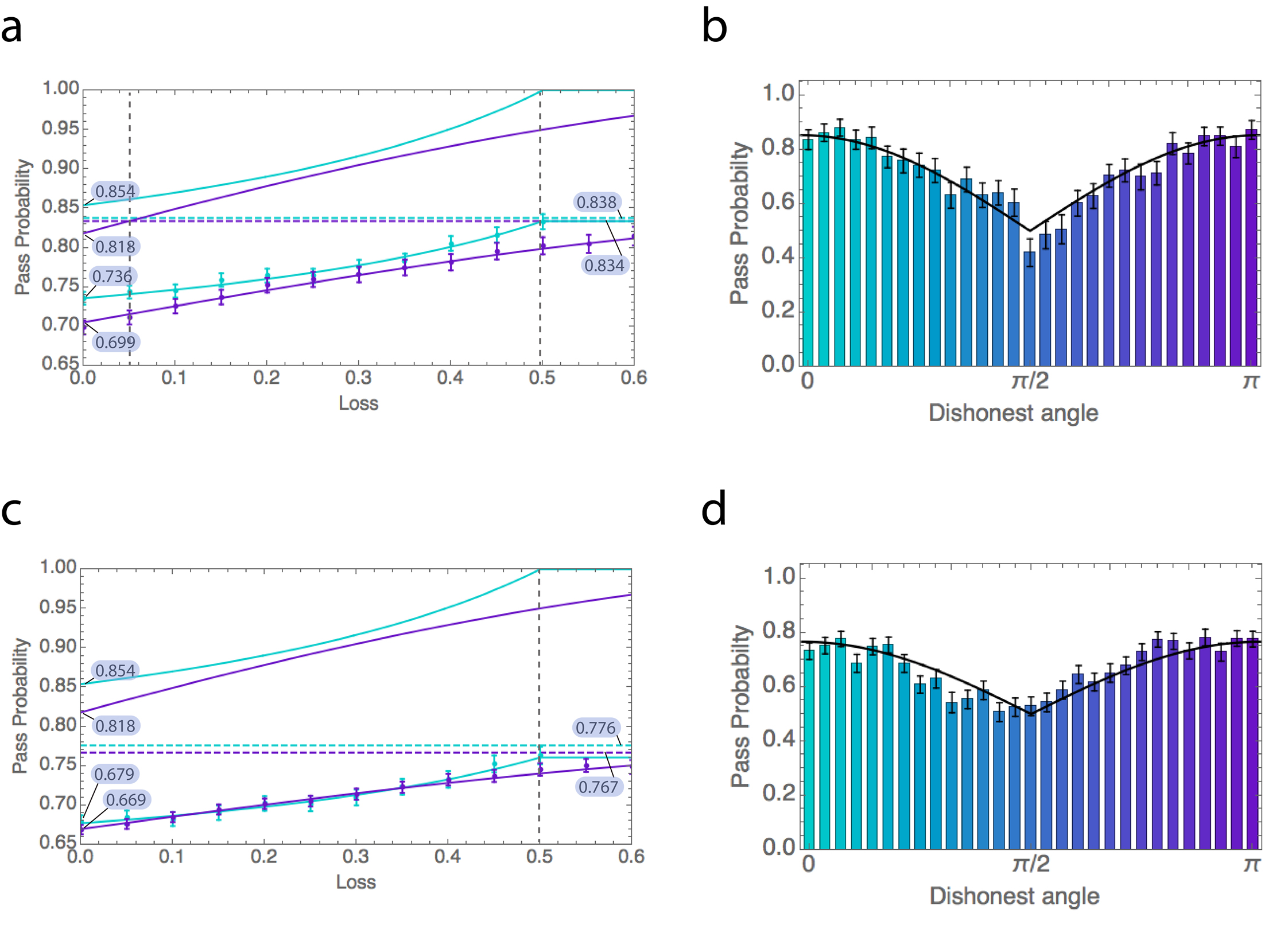}
\caption{Pass probabilities as a function of loss for one dishonest party in a three- and four-party setting. Panels {\bf a} and {\bf b} correspond to the three-party setting, and panels {\bf c} and {\bf d} correspond to the four-party setting. The upper curves in panels {\bf a} and {\bf c} show the ideal theoretical case for the GME bound for the $\theta$-protocol (purple curve) and a cheating strategy for the $XY$-protocol (turquoise curve) that always performs better. Note that the $XY$-protocol cannot be used here for verification as the non-GME dishonest value is always above the honest value. The lower solid curves in panels {\bf a} and {\bf c} correspond to the experimental results obtained for the three- and four-photon GHZ state, respectively. In both panels, the dashed lines correspond to the honest experimental values when there is no loss (turquoise for the $XY$-protocol and purple for the $\theta$-protocol). Panels {\bf a} and {\bf c} clearly show that the $\theta$-protocol can tolerate loss $\gtrsim 0.5$ in the ideal case. Panels {\bf b} and {\bf d} show the optimal pass probability that the dishonest party can obtain when running the $\theta$-protocol with no loss, for a given dishonest angle $\theta$, for the three-party and four-party case, respectively. In all plots the curves are a best fit to the data. All error bars represent the standard deviation and are calculated using a Monte Carlo method with Poissonian noise on the count statistics~\cite{James01}.}
\label{fig:LossDependentPassRates}
\end{figure*}

{\it 1. Verification of three-party GHZ --} The $XY$ verification protocol was initially carried out using the three-photon GHZ state, with all parties behaving honestly. The first two angles $\theta_j$ were randomly chosen to be either $0$ or $\pi/2$, with the third angle representing the Verifier being decided so that $\sum_j\theta_j$ is a multiple of $\pi$. After repeating the protocol on 6000 copies of the state, the pass probability was found to be $0.838 \pm 0.005$. Similarly, the $\theta$-protocol was carried out, with the first two angles chosen uniformly at random from the continuous range $[0,\pi)$. After 6000 copies of the state were prepared and measured, the pass probability was found to be $0.834 \pm 0.005$.

Using the relation between the fidelity and the pass probability, $F(\rho)\geq 2P(\rho)-1$, the Verifier can conclude that the fidelity with respect to an ideal GHZ state is at least $0.676 \pm 0.010$ for the $XY$-protocol and at least $0.668 \pm 0.010$ for the $\theta$-protocol. These values are consistent with the value obtained using state tomography. Despite the non-ideal experimental resource, the lower bound on the fidelity is clearly above 1/2 and therefore sufficient for the Verifier to verify GME in this all honest case.

More importantly, the $\theta$-protocol enables the Verifier to verify GME even when they do not trust all of the parties. Indeed, the experimental value of the pass probability, $0.834$, exceeds by more than three standard deviations the GME bound of $0.818$ for the dishonest case. We remark that for verifying GME in these conditions we crucially used the fact that our three-qubit GHZ state has very high fidelity and that the $\theta$-protocol has improved tolerance to noise. In fact, the Verifier is not able to verify GME using the $XY$-protocol, since the experimental value of  $0.838$ does not exceed the GME bound of $0.854$.

{\it 2. Verification of three-party GHZ with loss (theory) --} We now investigate the impact of loss on the performance of the verification protocols. In this setting, the Verifier is  willing to accept up to a certain loss rate from each party. When a party declares loss, the specific run of the protocol is aborted and the Verifier moves on to testing the next copy of the resource state. A dishonest party, who may not have the maximum allowed loss rate in their system, or may even have no loss at all, can increase the overall pass probability of the state by declaring loss whenever the probability to pass a specific measurement request from the Verifier is low. 

For example, a non-GME state can have pass probability 1 for the $XY$-protocol when the allowed loss rate is $50\%$. In this case, the source can share a state of the form $\frac{1}{\sqrt{2}}(\ket{HH}+\ket{VV})\otimes\ket{+}$, where the third qubit is sent to a dishonest party. Then, when the latter is asked to measure in the Pauli $X$ basis, he always answers correctly, while when asked to measure in the Pauli $Y$ basis he declares loss. Of course, such a strategy would alert the Verifier that the party is cheating, since he is always declaring loss when asked to measure in the $Y$ basis, while when asked to measure in the $X$ basis, he always measures the $|+\rangle$ eigenstate. However, if the source and the dishonest party are collaborating, and the source is able to create and share any Bell pair with the two honest parties, then the test can be passed each time without the cheating detected. The dishonest strategy would go as follows: the source sends randomly one of the four states $\{\frac{1}{\sqrt{2}}(\ket{HH}+\ket{VV}),\frac{1}{\sqrt{2}}(\ket{HH}-\ket{VV}),\frac{1}{\sqrt{2}}(\ket{HH}+i\ket{VV}),\frac{1}{\sqrt{2}}(\ket{HH}-i\ket{VV})\}$ and tells the dishonest party which one was sent, so that the latter can coordinate his actions. For the first state he replies 0 only for the $X$ basis, for the second state he replies 1 only for the $X$ basis, for the third he replies 1 only for the $Y$ basis and for the fourth he replies 0 only for the $Y$ basis.

More generally, we can analytically find the GME bound as a function of the loss rate for both protocols and describe optimal cheating strategies to achieve these bounds with non-GME states. The optimal cheating strategy for the $XY$-protocol consists of the source rotating the non-GME state that is sent to the honest parties in a specific way depending on the amount of loss allowed, and informing the dishonest party about the rotation. For zero loss, the optimal state is the $\pi/4$-rotated Bell pair $\frac{1}{\sqrt{2}}(\ket{HH}+e^{i\frac{\pi}{4}}\ket{VV})$, while for $50\%$ loss, the optimal state is the Bell pair $\frac{1}{\sqrt{2}}(\ket{HH}+\ket{VV})$. For any loss, $\lambda$, in between, the dishonest strategy is a probabilistic mixture of these two strategies; it consists of sending the Bell pair with probability $2\lambda$ (and discarding the rounds in which the dishonest party is asked to measure $Y$), and the $\pi/4$-rotated Bell pair with probability $1-2\lambda$. In both, the strategy mentioned in the previous paragraph for avoiding detection of the dishonest party's cheating is required. On the other hand, the optimal strategy for the $\theta$-protocol is having the source send a rotated Bell pair with the dishonest party declaring loss for the angles that have the lowest pass probability (see Supplementary Note 1). 

The upper bounds of the pass probability for the optimal cheating strategies using a non-GME state are shown as the solid turquoise and purple upper curves in Fig.~\ref{fig:LossDependentPassRates}, for the $XY$ and $\theta$-protocol respectively. Specifically for the case of no loss, we recover the GME bounds of $0.854$ and $0.818$ for the $XY$- and $\theta$-protocol, respectively. The GME bound for the $XY$-protocol reaches 1 for $50\%$ loss, while the GME bound for the $\theta$ protocol reaches 1 only at $100\%$ loss. 

{\it 3. Verification of three-party GHZ with loss (experiment) --} In Fig.~\ref{fig:LossDependentPassRates}a one can see the experimental value of $0.834 \pm 0.005$ when there is no loss for the $\theta$-protocol enables the Verifier to verify GME in the presence of up to $\sim 5\%$ loss -- once the loss increases past $5\%$, the Verifier can no longer guarantee the shared experimental state has GME. Again, this loss tolerance is only possible due to the high fidelity of our three-party GHZ state and the fact that our $\theta$-protocol has a better behaviour with respect to loss. The tolerance to loss can be further improved using experimental states with higher fidelities. However, it is interesting to note that $5\%$ loss corresponds to $\sim$1 km of optical fibre, which already makes the protocol relevant to a quantum network within a small area, such as a city or government facility, where a number of quantum communication protocols could be carried out over the network, such as, for instance quantum secret sharing~\cite{Tittel01}, telecloning~\cite{Radmark} and open destination teleportation~\cite{Zhao}.

{\it 4. Implementation of dishonest strategies for three-party GHZ --} In order to maximise the pass probabilities of the protocols using a non-GME state, the source needs to appropriately rotate the state that is sent to the honest parties depending on the amount of loss allowed. We implemented this strategy for a single dishonest party by using a complementary method, where the source creates a three-qubit GHZ state and gets the dishonest party to perform a projective measurement that creates the necessary rotated non-GME state between the honest parties. This strategy was performed experimentally for both protocols on 3000 copies of the three-qubit GHZ state. Since in our experiment the GHZ states are created by postselection, the loss corresponds to the allowed percentage of tests in which the dishonest party can claim they lost their qubit during transmission of the corresponding photon from the source.

The pass probabilities are shown as a function of loss by the solid turquoise and purple lower curves in Fig.~\ref{fig:LossDependentPassRates}a. They show the same trend as the previous curves but are shifted lower due to the non-ideal experimental state. For the no loss case, we obtain a pass probability of $0.736 \pm 0.008$ for the $XY$-protocol. For the $\theta$-protocol, the pass probability depends on the dishonest party's measurement request $\theta$: for no loss, the experimental results are shown in Fig.~\ref{fig:LossDependentPassRates}b, from which we obtain an average pass probability of $0.699 \pm 0.009$. When loss is included the dishonest party's cheating strategy leads to a higher pass probability, since the dishonest party claims loss when the angle given to him by the Verifier is close to $\pi/2$, corresponding to the minimum pass probability shown in Fig.~\ref{fig:LossDependentPassRates}b. Similar to the discussion in the example of the $XY$-protocol, the source collaborates with the dishonest party and applies a rotation to the shared state, so that the declared lost angles appear uniform and not always around $\pi/2$.

{\it 5. Verification of four-party GHZ --} To check the performance of the protocols for a higher number of parties, the verification tests were carried out using the four-photon GHZ state generated in our experiment, now with three angles chosen randomly, and the fourth depending on the condition that $\sum_j\theta_j$ is a multiple of $\pi$.  Again, we start with the all honest case where any of the parties may be the Verifier. For the $XY$-protocol, with all $\theta_j$ equal to $0$ or $\pi/2$, the pass probability for 6000 copies of the state was found to be $0.776 \pm 0.005$. For the $\theta$-protocol, using 6000 copies, the pass probability was found to be $0.767 \pm 0.005$.

As in the three-party case, the Verifier can conclude that the fidelity with respect to an ideal GHZ state is at least $0.552 \pm 0.010$ for the $XY$-protocol and at least $0.534 \pm 0.010$ for the $\theta$-protocol, therefore just sufficient for the Verifier to verify that GME is present in the state. Again, the high fidelity of our experimental state is crucial for this result. Nevertheless, none of the two protocols can confirm GME in the presence of dishonest parties since the pass probabilities are below the GME bounds of $0.854$ and $0.818$, respectively.

{\it 6. Implementation of dishonest strategies for four-party GHZ --} The dishonest strategies that are used to implement the two verification protocols for different amounts of loss are the same as in the three-party case. However, we proceed in two different ways for a single dishonest party. First, we have the source create our non-ideal four-qubit GHZ state and then allow the dishonest party to perform the dishonest projective measurement in order to create a non-GME state. When there is no loss we obtain a pass probability of $0.679 \pm 0.008$ for the $XY$-protocol and $0.669 \pm 0.008$ for the $\theta$-protocol (averaged over the dishonest angle $\theta$, as shown in the histogram of Fig.~\ref{fig:LossDependentPassRates}d). When loss is included, the pass probabilities of both the $XY$- and $\theta$-protocols increase, as the dishonest party uses the loss to their advantage (see Fig.~\ref{fig:LossDependentPassRates}c). A second way to implement the dishonest strategy is to have the source create the non-ideal three-qubit GHZ state for the honest parties and the dishonest party hold an unentangled photon. This results in a four-party non-GME state with reduced noise -- as the dephasing from the entangled pair of the second PCF is no longer present~\cite{Bell12}. We perform the $\theta$-protocol with this better quality resource state and see that the pass probability increases from $0.669 \pm 0.005$ to $0.698 \pm 0.008$ for the no loss case and remains higher when loss is included (see Fig.~\ref{fig:GHZ4vsGHZ3}). Note that despite the second strategy having higher pass probabilities, these are still below the GME bound shown in Fig.~\ref{fig:LossDependentPassRates}c (upper purple curve). 

\begin{figure}[t]
\includegraphics[width=7.5cm]{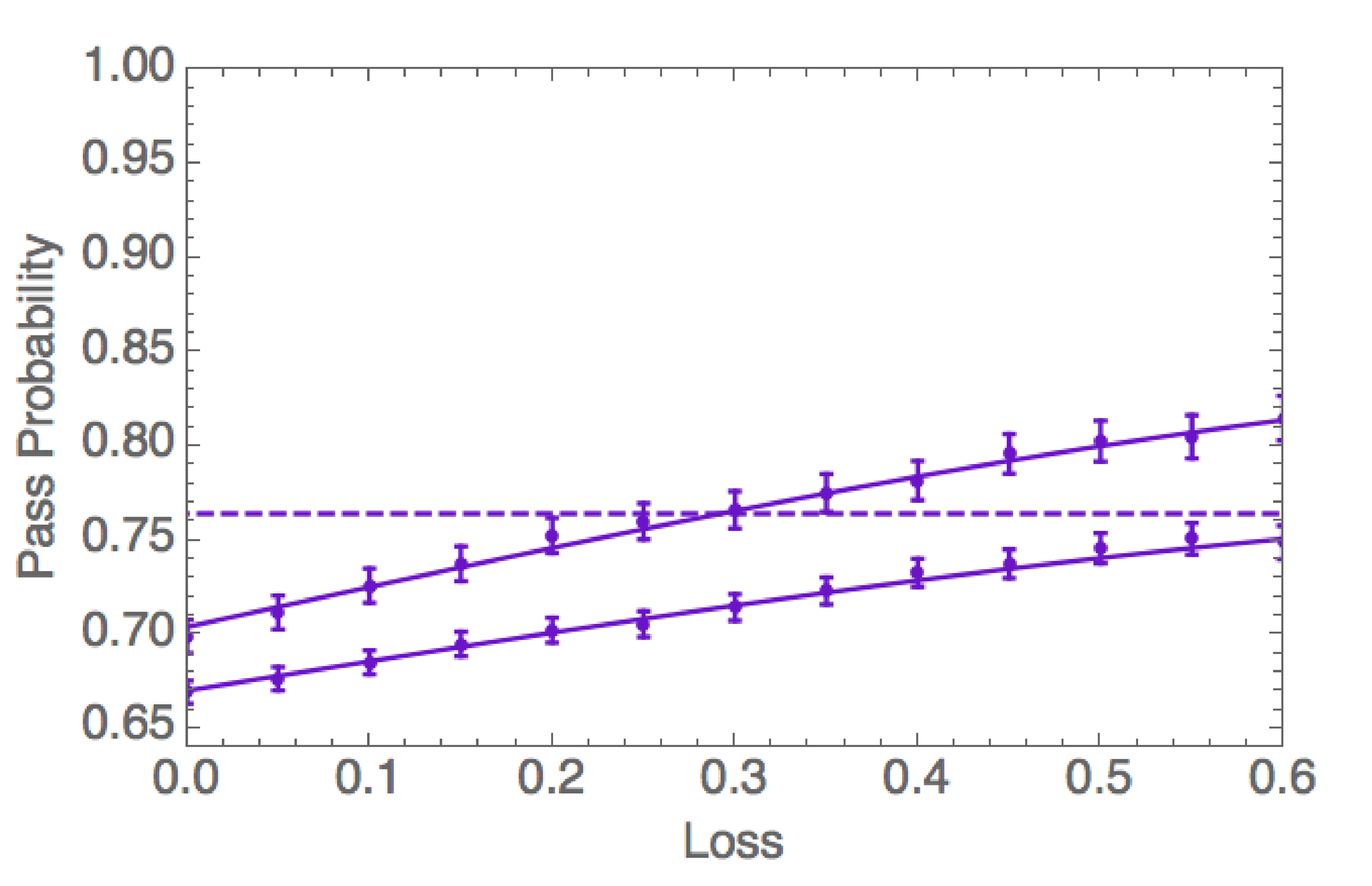}
\caption{Impact of noise and loss on the pass probability of the $\theta$-verification protocol in a four-party setting. The lower curve corresponds to a biseparable state (four-qubit GHZ state with a projective measurement on the dishonest qubit) and the upper curve corresponds to a biseparable state (three-qubit GHZ state and an unentangled qubit for the dishonest player) that has less noise. The dashed line corresponds to the honest case. All error bars represent the standard deviation and are calculated using a Monte Carlo method with Poissonian noise on the count statistics~\cite{James01}.}
\label{fig:GHZ4vsGHZ3}
\end{figure}

The comparison of the two strategies shows that the projection method is not necessarily optimal for the dishonest party due to phase noise in the experimental state. Note also that as the pass probability of the experimental state in the honest case (dotted purple line in Fig.~\ref{fig:LossDependentPassRates}c) is below the GME bound, the Verifier is not able to verify GME for this four-party setting for any amount of loss. Verification of GME is achieved in our experiment only in the three-party setting. However, four-party verification could be achieved using experimental states with higher fidelities, and even with our non-ideal three-party GHZ state we have been able to provide the first proof-of-principle demonstration of our GME verification protocol.

\section{Discussion}

The results we have presented are situated in a realistic context of distributed communication over photonic quantum networks: we have shown that it is possible for a party in such a network to verify the presence of genuine multipartite entanglement in a shared resource, even when some of the parties are not trusted, including the source of the resource itself. This distrustful setting sets particularly stringent conditions on what can be shown in practice. With our state-of-the-art optical setup that produces high-fidelity three- and four-photon GHZ states, we were able to show, for the three-party case, that this verification process is possible using a carefully constructed protocol, for up to 5\% loss, under the most strict security conditions. Clearly, the loss tolerance of the system can be further improved by using states with even higher fidelities. This would also enable the implementation of the verification protocols for a larger number of qubits.

It is important to remark that our verification protocols go beyond merely detecting entanglement; they also link the outcome of the verification tests to the state that is actually used by the honest parties of the network with respect to their ideal target state. This is non trivial and of great importance in a realistic setting where such resources are subsequently used by the parties in distributed computation and communication applications executed over the network. Such applications may also require multipartite entangled states other than the GHZ states studied in this work. We expect that our verification protocols should indeed be applicable to other types of useful states such as, for instance, stabiliser states.

\section{Acknowledgments}

This work was supported by the UK's Engineering and Physical Sciences Research Council, ERC grants 247462 QUOWSS and QCC, EU FP7 grant 600838 QWAD, the Ville de Paris Emergences project CiQWii, the ANR project COMB, the Ile-de-France Region project QUIN and the South African National Research Foundation.



\newpage

\setcounter{equation}{0}
\setcounter{figure}{0}

\renewcommand{\figurename}{FIG. S}
\renewcommand{\thefigure}{\arabic{figure}}

\begin{widetext}


\section{Supplementary Material}

\subsection{Supplementary Figures}

\begin{figure}[h]
\centering
\includegraphics[width=9cm]{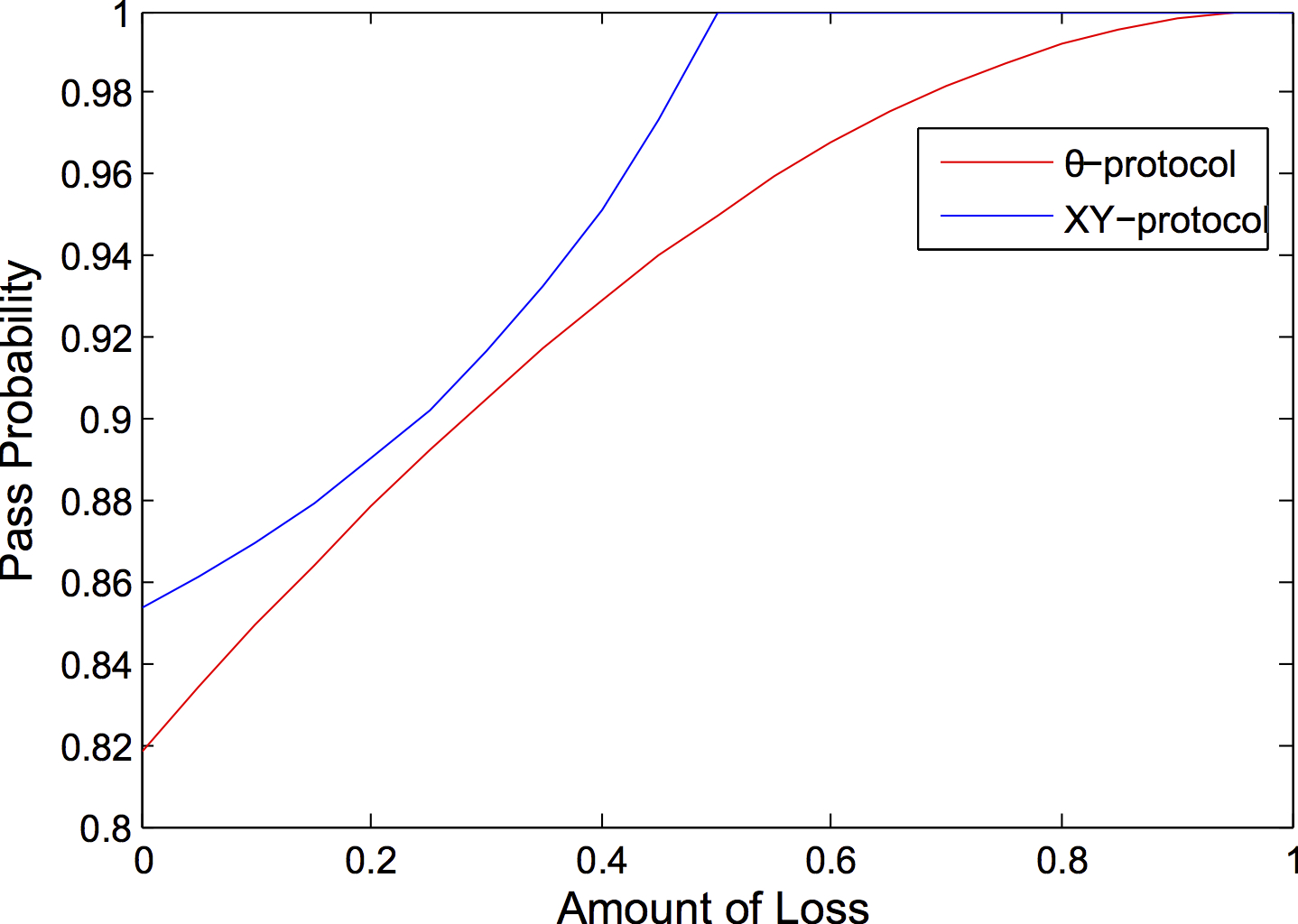}
\caption{Loss tolerance of the original and new verification protocols. The $\theta$-protocol test performs better than the $XY$-protocol test and is still viable after $50 \%$ loss.}
\label{fig:GHZ4vsGHZ3LossDependence}
\end{figure}

\begin{figure}[h]
\centering
\includegraphics[width=13cm]{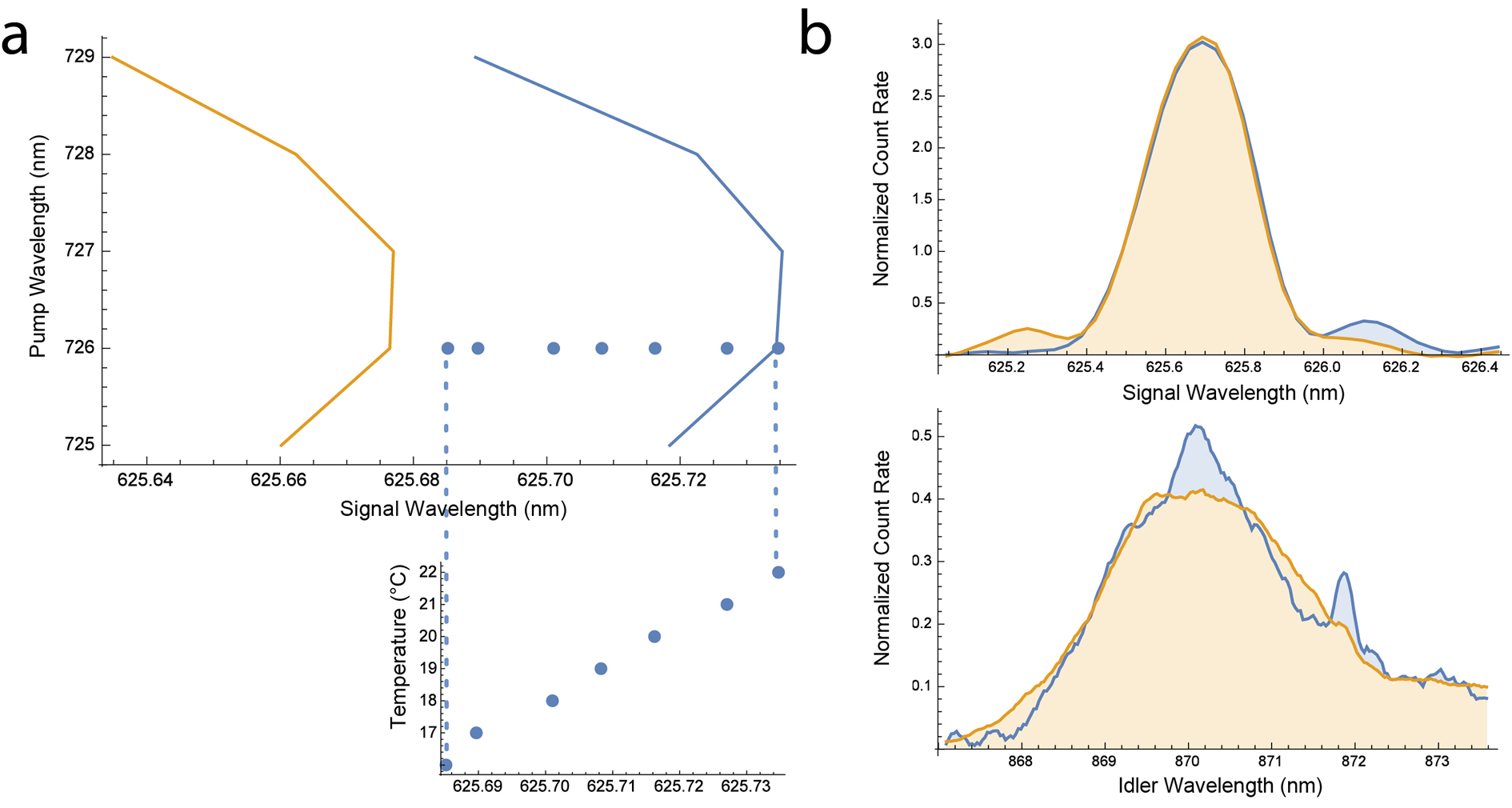}
\caption{Tuning the photonic crystal fibre sources. Spectra of the microstructured photonic crystal fiber sources. {\bf a,} Central wavelengths for source 1 (orange line) and source 2 (blue line) with varying pump wavelength. The temperature tuning of source 2 (blue points) is also shown. {\bf b,} Spectra of signal photons (top) and idler photons (bottom) from source 1 (orange) and source 2 (blue) tuned at $23.7^{\circ}C$ with a $726$~nm pump.}
\label{fig:SpectraPlots}
\end{figure}

\section*{Supplementary Note 1}

\paragraph{}In this section we provide further details of the $\theta$-protocol presented in the main text (see Fig.~1). Note that the proofs for the $XY$-protocol follow easily as a special case with only two measurement settings.


\subsection*{Correctness of the $\theta$-protocol}

\paragraph{}Here we prove that the $n$-qubit GHZ state passes the verification test with probability 1. The measurements that the parties perform in the $X$-$Y$ plane are equivalent to rotation operators around the $Z$ axis of the Bloch sphere:
\begin{equation}
R_z(\theta_j)=\begin{bmatrix}
1 & 0 \\
0 & e^{-i\theta_j}
\end{bmatrix}
\end{equation}
followed by a measurement in the Pauli $X$ basis. After the application of the rotation operators $R_z(\theta_j)$ on an $n$-qubit GHZ state, we end up with the state $\frac{1}{\sqrt{2}}\big(\ket{0}^{\otimes n}+e^{-i\Theta}\ket{1}^{\otimes n} \big)$, where $\Theta=\sum_{j=1}^n\theta_j$. When $\Theta=0\pmod{2\pi}$ the shared state written in the Pauli $X$ basis is given by a linear summation of terms with an even number of $\ket{-}$ states for the parties. On the other hand when $\Theta=\pi\pmod{2\pi}$ the shared state is given by a linear summation of terms with an odd number of $\ket{-}$ states for the parties. Thus, when the parties measure their qubits in the Pauli $X$ basis the parity of their measurements will be zero if $\Theta=0\pmod{2\pi}$ and one if $\Theta=\pi\pmod{2\pi}$. In other words, the test will always be passed with unit probability.


\subsection*{Security in the Honest Model}

\paragraph{}Now we prove a lower bound for the fidelity of the shared state, when all parties are honest, that depends on the pass probability of the test $P(\rho)$. We will specifically prove the following theorem:

\begin{theorem}[Honest Case]\label{Theorem_hon}
Let $\rho$ be the state shared between $n$ parties. \\
If $F(\rho,\ket{G_0^n}):= \bra{G_0^n}\rho\ket{G_0^n}$, where $\ket{G_0^n}$ is an $n$-qubit GHZ state, then $F(\rho)\geq 2P(\rho)-1$.
\end{theorem}

\paragraph{}Let us define a test in order to verify a `rotated' GHZ state, namely $\ket{G_\Theta^n}=1/\sqrt{2}(\ket{0}^{\otimes n}+e^{-i\Theta}\ket{1}^{\otimes n})$, where $\Theta\in[0,2\pi)$. Here, the sum of the angles of the parties has to comply with the condition: $\sum_{j=1}^n\theta_j-\Theta  \equiv 0 \pmod{\pi}$. The test that we are interested in is the following:
\begin{equation}
\bigoplus_{j=1}^n Y_j=\frac{\sum_{j=1}^n\theta_j-\Theta}{\pi} \pmod 2
\end{equation} \newline
Let $\{P^n_\Theta, I-P^n_\Theta\}$ be the POVM that corresponds to the above test. We will prove by induction that:
\begin{equation}
P^n_\Theta=\ketbra{G^n_\Theta}{G^n_\Theta}+\frac{1}{2}I_{n}^\Theta
\end{equation}
where $I_{n}^\Theta$ is the projection on the space orthonormal to $\ket{G_\Theta^n}$ and $\ket{G_{\Theta+\pi}^n}$.

\paragraph{}For $n=1$ we have that $P^1_\Theta=\ketbra{G^1_{\theta_1}}{G^1_{\theta_1}}$ so the statement holds. We assume it is true for $n$ and we show the statement for $n+1$.

\paragraph{}Let $\{P^{n+1}_\Theta(\theta_1), I-P^{n+1}_\Theta(\theta_1)\}$ be the POVM that corresponds to the test for a given angle $\theta_1$. There are two cases:
\begin{enumerate}
\item Party 1 outputs $Y_1=0$. Then, the following equality should hold:
\begin{equation}
\bigoplus_{j=2}^{n+1} Y_j=\frac{\sum_{j=2}^{n+1}\theta_j-(\Theta-\theta_1)}{\pi} \pmod 2
\end{equation}
\item Party 1 outputs $Y_1=1$. Then, the following equality should hold:
\begin{equation}
\bigoplus_{j=2}^{n+1} Y_j=\frac{\sum_{j=2}^{n+1}\theta_j-(\Theta-\theta_1+\pi)}{\pi} \pmod 2
\end{equation}
\end{enumerate}

\paragraph{}Let $\Theta'\equiv\Theta-\theta_1\pmod{2\pi}$. It is evident that the first outcome of the test is equivalent to $P^n_{\Theta'}$ and the second to $I-P^n_{\Theta'}$. For any given $\theta_1$, we have:
\begin{eqnarray}
P^{n+1}_\Theta(\theta_1)&=&\ketbra{G^1_{\theta_1}}{G^1_{\theta_1}}\otimes P^n_{\Theta'}+\ket{G^1_{\theta_1+\pi}}\bra{G^1_{\theta_1+\pi}}\otimes(I-P^n_{\Theta'})\\
&=&\ketbra{G^1_{\theta_1}}{G^1_{\theta_1}}\otimes\ketbra{G^n_{\Theta'}}{G^n_{\Theta'}}+ \ket{G^1_{\theta_1+\pi}}\bra{G^1_{\theta_1+\pi}}\otimes \ket{G^n_{\Theta'+\pi}}\bra{G^n_{\Theta'+\pi}} \nonumber \\
&+&\frac{1}{2}\Big(\ketbra{G^1_{\theta_1}}{G^1_{\theta_1}}+\ket{G^1_{\theta_1+\pi}}\bra{G^1_{\theta_1+\pi}}\Big)\otimes I_n^{\Theta'} \nonumber \\
&=&\ketbra{G^{n+1}_\Theta}{G^{n+1}_\Theta}+\ketbra{\Phi_{\theta_1}}{\Phi_{\theta_1}}+\frac{1}{2}
I_1\otimes I_n^{\Theta'} \nonumber
\end{eqnarray}
where we define:
\begin{equation}
\ket{\Phi_a}=\frac{1}{\sqrt{2}}\big(\ket{G^1_a}\ket{G^n_{\Theta-a}}-\ket{G^1_{a+\pi}}\ket{G^n_{\Theta-a+\pi}}\big)
\end{equation}
It is straightforward to verify that:
\begin{eqnarray}
I_{n+1}^{\Theta}&=&\ketbra{\Phi_{\theta_1}}{\Phi_{\theta_1}}+\ketbra{\Phi_{\theta_1+\frac{\pi}{2}}}{\Phi_{\theta_1+\frac{\pi}{2}}}+I_1\otimes I_n^{\Theta'}
\end{eqnarray}
where as before $I_{n+1}^\Theta$ is the projection on the space orthonormal to $\ket{G^{n+1}_\Theta}$ and $\ket{G^{n+1}_{\Theta+\pi}}$.
Since angle $\theta_1$ is chosen uniformly at random in $[0,\pi)$, we have that:
\begin{eqnarray}
P^{n+1}_\Theta&=&\frac{1}{\pi}\int_0^{\pi}P_\Theta^{n+1}(\theta_1)d\theta_1 \\
                          &=&\frac{1}{\pi}\int_0^{\pi/2}\Big[P_\Theta^{n+1}(\theta_1)+P_\Theta^{n+1}(\theta_1+\frac{\pi}{2})\Big]d\theta_1\nonumber\\
                           &=&\ketbra{G^{n+1}_\Theta}{G^{n+1}_\Theta}+\frac{1}{2}I_{n+1}^\Theta\nonumber
\end{eqnarray}
For $\Theta=0\pmod {2\pi}$, we can easily infer the basic argument of the proof, that the test is equivalent to performing the POVM \{$P^n_0,I-P^n_0\}$. We can therefore express any state $\rho$ with fidelity $F(\rho)$ to the GHZ state as
$\rho=F(\rho)\ketbra{G_0^n}{G_0^n}+(1-F(\rho))\mathcal{\chi},
$
where $\mathcal{\chi}$ is a $2^n\times 2^n$ density matrix with zero in the place of $\ketbra{G_0^n}{G_0^n}$. We then have $P(\rho) =   \tr(P_0^n \rho) \leq \frac{1}{2}+\frac{F(\rho)}{2}$.\\


\subsection*{Security in the Dishonest Model}

{\bf Figures of merit for the dishonest case.} Without loss of generality, the source generates a state $\sum_r p_r \kb{r} \otimes \kb{\Psi_r}_{\mathcal{HDE}}$ where $r$ corresponds to some classical information controlled by the dishonest players, and $\mathcal{HDE}$ are respectively the Hilbert space of the honest parties, the dishonest parties and the external environment, which no parties can control.

\paragraph{}Here we prove a lower bound for the fidelity of the shared state, when the $n-k$ parties are dishonest and there are no loss in the system. Since we consider that the dishonest parties can collaborate between themselves and with the source, any security statement should consider that they can apply any operation $U_r$ (possibly depending on $r$) to their part of the state that works to their advantage. More specifically, we prove the following theorem:

\anote{\begin{theorem}[Dishonest Case]\label{Theorem_dis'}
		Let $\rho = \sum_{r=1}^{R} p_r \kb{r} \otimes \rho_r$ be the state shared between $n$ parties in the space $\mathcal{HD}$. If $F'(\rho):= \sum_r p_r \max_{\hspace*{0.05cm} U^r_{n-k}} F\big((\mathbb{I}_k \otimes U^r_{n-k}) \rho_r (\mathbb{I}_k\otimes (U^r_{n-k})^\dag),\ket{G_0^n}\big)$, where $U^r_{n-k}$ are operators on the space of the dishonest parties, then $F'(\rho)\geq 4 P(\rho)-3$.
	\end{theorem}}
	
	\begin{proof}
		\anote{\paragraph{Case 1 : Pure state.} We first consider the case without classical information $r$ and without environment, $\ie$ where $\rho$ is a pure state $\kb{\Psi}_{\mathcal{HD}}$. We write }
\be		
\ket{\Psi}=\ket{G_\theta^k}\ket{\Psi_\theta}+\ket{G_{\pi+\theta}^k} \ket{\Psi_{\pi+\theta}}+\ket{\mathcal{X}}
\ee
		where $\theta=\sum_{j\in H} \theta_j\pmod{\pi}$ is the honest angle, $H$ is the set of the honest parties  and $\ket{G_\alpha^k}=1/\sqrt{2}(\ket{0}^{\otimes k}+e^{i\alpha}\ket{1}^{\otimes k})$ for any angle $\alpha$. Note that the component of the honest parties in $\ket{\mathcal{X}}$ is orthogonal to both  $\ket{G_\theta^k}$ and $\ket{G_{\pi+\theta}^k}$.
		
		\paragraph{}The dishonest parties want to know in which of the two states $\ket{G^k_\theta}$ and $\ket{G^k_{\theta+\pi}}$ the state the honest parties share will collapse into after the measurement, and by consequence what will be the honest output $Y_H=\bigoplus_{i\in H}Y_i$. They will perform a Helstrom measurement on their share in order to distinguish between $\ket{\Psi_\theta}$ and $\ket{\Psi_{\theta+\pi}}$. This measurement is optimal and gives the following bound:
		\be
			\text{Pr[guess }Y_H|\theta]=\frac{1}{2}+\frac{1}{2}\Big\|\ketbra{\Psi_{\theta}}{\Psi_{\theta}}-\ketbra{\Psi_{\theta+\pi}}{\Psi_{\theta+\pi}}\Big\|
		\ee
		
		\paragraph{}To calculate the above norm, we make use of a known property, that the trace norm of a Hermitian matrix is equal to the sum of the absolute values of its eigenvalues. After some simple calculations we can verify that the above probability is equal to:
		\begin{align}\label{cheat_prob}
			\text{Pr[guess }Y_H|\theta] &=\frac{1}{2}+\frac{1}{2}\sqrt{\big(\norm{\ket{\Psi_\theta}}^2+\norm{\ket{\Psi_{\theta+\pi}}}^2 \big)^2-4|\braket{\Psi_\theta}{\Psi_{\theta+\pi}}|^2}\nonumber\\
			&\leq\frac{1}{2}+\frac{1}{2}\Big(\frac{\big(\norm{\ket{\Psi_\theta}}^2+\norm{\ket{\Psi_{\theta+\pi}}}^2 \big)^2-4|\braket{\Psi_\theta}{\Psi_{\theta+\pi}}|^2+1}{2}\Big)\nonumber\\
			&=\frac{3}{4}+\frac{1}{4}\Big(\big(\norm{\ket{\Psi_\theta}}^2+\norm{\ket{\Psi_{\theta+\pi}}}^2 \big)^2 -4|\braket{\Psi_\theta}{\Psi_{\theta+\pi}}|^2\Big)
		\end{align}
		
		\paragraph{}We now perform a Schmidt decomposition of $\ket{G^k_\theta}\ket{\Psi_\theta} + \ket{G^k_{\theta + \pi}}\ket{\Psi_{\theta + \pi}}$:
		\begin{equation}\label{4}
			\ket{G^k_\theta}\ket{\Psi_\theta} + \ket{G^k_{\theta + \pi}}\ket{\Psi_{\theta + \pi}} = \ket{A_\theta^0}\ket{B_\theta^0} + \ket{A_\theta^1}\ket{B_\theta^1} \end{equation}
		where $\braket{A_\theta^0}{A_\theta^1} = \braket{B_\theta^0}{B_\theta^1} = 0$. We use the following normalization: $||\ket{A_\theta^0}||^2 = ||\ket{A_\theta^1}||^2 = 1$, $||\ket{B_\theta^0}||^2 = p_\theta$, $||\ket{B_\theta^1}||^2 = q_\theta$. There exist $z_0,z_1\in\mathbb{C}$ such that:
		\be
		\ket{A_\theta^0} = z_0 \ket{G^k_\theta} + z_1 \ket{G^k_{\theta + \pi}} \quad \textrm{and} \quad
		\ket{A_\theta^1} = z_1^* \ket{G^k_\theta} - z_0^{*} \ket{G^k_{\theta + \pi}}
		\ee
		where $|z_0|^2 + |z_1|^2 = 1$, which gives us:
		\begin{align}
			\ket{A_\theta^0}\ket{B_\theta^0} + \ket{A_\theta^1}\ket{B_\theta^1}&= (z_0 \ket{G^k_\theta} + z_1 \ket{G^k_{\theta + \pi}})\ket{B_0} + (z_1^* \ket{G^k_\theta} - z_0^{*} \ket{G^k_{\theta + \pi}}) \ket{B_1} \nonumber \\
			&= \ket{G^k_{\theta}}(z_0 \ket{B_\theta^0} + z_1^* \ket{B_\theta^1}) + \ket{G^k_{\theta + \pi}}(z_1 \ket{B_\theta^0} - z_0^* \ket{B_\theta^1})
		\end{align}
		and from Eq.~(\ref{4}) we have:
		\begin{align}
			\ket{\Psi_\theta} = z_0 \ket{B_\theta^0} + z_1^* \ket{B_\theta^1} \nonumber \\
			\ket{\Psi_{\theta + \pi}} =  z_1 \ket{B_\theta^0} - z_0^* \ket{B_\theta^1}
		\end{align}
		Since $\ket{A_\theta^0}$ and $ \ket{A_\theta^1}$ are on the same subspace as $ \ket{G^k_\theta}$ and $\ket{G^k_{\theta + \pi}}$, there exist $x\in\mathbb{R},y\in\mathbb{C}$ such that:
\be
		\ket{A_\theta^0} = x \ket{0^k} + y \ket{1^k} \quad \textrm{and} \quad \ket{A_\theta^1} = y^* \ket{0^k} - x \ket{1^k}  
\ee
		where  $x^2 + |y|^2 = 1$ (we can assume that  $x \in \mathbb{R}$ up to a global phase on $\ket{A_0}$ and $\ket{A_1}$). Then:
		\begin{align}
			|z_0|^2 = |\braket{A_\theta^0}{G_k^\theta}|^2 = \frac{1}{{2}}|x + ye^{i \theta}|^2
		\end{align}
		and since $y\in\mathcal{C}$, we rewrite $y = |y|e^{i \alpha}$ and get:
		\begin{align}
			|z_0|^2 & = \frac{1}{2} \big|x + |y|e^{i \theta + \alpha}\big|^2
			= \frac{1}{2}(1 + 2x|y|\cos(\theta + \alpha))
		\end{align}
		Using $|z_0|^2 + |z_1|^2 = 1$, we have $|z_1|^2 = \frac{1}{2}(1 - 2x|y|\cos(\theta + \alpha))$. Also, from $x^2,|y|^2 \ge 0$ and $x^2 + |y|^2 = 1$, we have that $x^2|y|^2 \le 1/4$. This gives us:
		\begin{align}\label{inner_prod}
			|\braket{\Psi_\theta}{\Psi_{\theta + \pi}}|^2 &= (p_\theta-q_\theta)^2 |z_0|^2 |z_1|^2 = (p_\theta-q_\theta)^2 \frac{1}{4}(1 - 4x^2|y|^2\cos^2(\theta + \alpha))\nonumber\\
			&\ge (p_\theta-q_\theta)^2 \frac{1}{4}(1 - \cos^2(\theta + \alpha))= (p_\theta-q_\theta)^2 \frac{1}{4}\sin^2(\theta + \alpha)
		\end{align}
		We then revisit Eq.~(\ref{cheat_prob}):
		\begin{equation}\label{guessoutput}
			\Pr[\text{guess }Y_H | \theta] \le \frac{3}{4}+ \frac{1}{4}\big((p_\theta+q_\theta)^2 - (p_\theta-q_\theta)^2 \sin^2(\theta+\alpha)\big)
		\end{equation}
		
		\paragraph{}Now, let us consider the optimal local operation that the dishonest parties can perform on their state, in order to maximize their cheating probability. If the reduced density matrices of the honest parties of the ideal state $\ket{G_0^n}$ and the state $\rho$ are $\sigma_H$ and $\rho_H$ respectively, it holds that there exists a local operation $R$ on the dishonest state that maximizes the fidelity:
		\be
		F'(\rho)=F((I\otimes R)\ket{\Psi},\ket{G_0^n}) = F(\sigma_H,\rho_H)
		\ee
		
		\paragraph{}Let us decompose $\ket{G_0^n}$ in the same orthonormal bases for the honest parties, as we did for $\ket{\Psi}$. We have  $\ket{G_0^n}=\ket{A_\theta^0}\ket{C^0}+\ket{A_\theta^1}\ket{C^1}$. Then:
		\begin{align}
			\sigma_H&=\frac{1}{2}\big(\ketbra{A_\theta^0}{A_\theta^0}+\ketbra{A_\theta^1}{A_\theta^1}    \big)\\
			\rho_H   &=p_\theta\ketbra{A_\theta^0}{A_\theta^0}+q_\theta\ketbra{A_\theta^1}{A_\theta^1}+\tr_{n-k}{\ket{\mathcal{X}}\bra{\mathcal{X}}}
		\end{align}
		and we can express fidelity $F'(\rho)=\tr[\sqrt{\sqrt{\rho_H}\sigma_H\sqrt{\rho_H}}]^2$, which gives:
		\begin{eqnarray}\label{fidelity'}
			F'(\rho)&=&\frac{1}{2}(\sqrt{p_\theta}+\sqrt{q_\theta})^2=\frac{p_\theta+q_\theta}{2}+\sqrt{p_\theta q_\theta} \nonumber \\
			&\geq&\frac{(p_\theta+q_\theta)^2}{2}+2p_\theta q_\theta=(p_\theta+q_\theta)^2-\frac{(p_\theta-q_\theta)^2}{2}
		\end{eqnarray}
		because for all non-negative $p$ and $q$ such that $p+q \le 1$, it holds that $p+q\geq (p+q)^2$ for $p+q \le 1$ and also that $\sqrt{pq}\geq 2pq$. Let us note here that whatever decomposition we do to the state $\ket{\Psi}$, the sum $(p_\theta+q_\theta)$ is a constant that always equals $\|\ket{\Psi_\theta}\|^2+\|\ket{\Psi_{\theta+\pi}}\|^2$. It follows that $(p_\theta-q_\theta)^2$ is lower bounded by the constant $2((p_\theta+q_\theta)^2-F'(\rho))$. Since $\theta$ is chosen uniformly at random, we have that:
		\begin{align}
			P(\rho)&= \frac{1}{\pi} \int_{0}^{\pi} \Pr[\text{guess }Y_H| \theta] \\
			&\leq  \frac{3}{4}+ \frac{1}{4}\Big((p_\theta+q_\theta)^2 - \frac{1}{\pi}\int_{0}^{\pi}(p_\theta-q_\theta)^2\sin^2(\theta + \alpha) d\theta\Big) \\
			& \leq \frac{3}{4}+ \frac{1}{4}\Big((p_\theta+q_\theta)^2+ F'(\rho)- (p_\theta+q_\theta)^2\Big) \\
			&\leq \frac{3}{4}+ \frac{1}{4}F'(\rho)
		\end{align}
		\anote{
			\paragraph{Case 2 : No classical information, mixed state.} We consider the case where $\rho = \sum_j q_j \kb{\Psi_j}_{\mathcal{HD}}$. Since the two functions $P(\cdot)$ and $F(\cdot)$ are linear, we can write 
			\begin{align}
				P(\rho) = \sum_j q_j P(\kb{\Psi_j} \le \frac{3}{4} + \frac{1}{4}\sum_j q_j F'(\kb{\Psi_j}) = \frac{3}{4} + \frac{1}{4}F'(\rho)
			\end{align}
			\paragraph{Case 3 : General Case.} We write $\rho = \sum_{r=1}^R p_r \kb{r} \otimes \rho_r$. We then write
			\begin{align}
				P(\rho) & = \sum_r p_r P(\rho_r) = \frac{3}{4} + \frac{1}{4} \sum_r p_r(F(\rho_r)) \\
				& = \frac{3}{4} + \frac{1}{4} \sum_r p_r \max_{\hspace*{0.05cm} U^r_{n-k}} F\big((\mathbb{I}_k \otimes U^r_{n-k}) \rho_r (\mathbb{I}_k\otimes (U^r_{n-k})^\dag),\ket{G_0^n}\big)
			\end{align} 
		}
	\end{proof}

\begin{corollary}
Let $\rho$ be the state shared between $n$ parties. If $F'(\rho):=\max_U F\big((\mathbb{I}_k\otimes U_{n-k})\rho(\mathbb{I}_k\otimes U_{n-k}),\ket{G_0^n}\big)=\frac{1}{2}$, where $U$ is an operator on the space of the dishonest parties, then

\begin{enumerate}
\item if the parties run the $\theta$-protocol, $P(\rho|\theta\text{-protocol})\leq\frac{1}{2}+\frac{1}{\pi}$.
\item if the parties run the $XY$-protocol, $P(\rho|XY\text{-protocol})\leq \cos^2(\frac{\pi}{8})$.
\end{enumerate}
\end{corollary}

\begin{proof}
We will first show the upper bound of the pass probability for the $\theta$-protocol and then examine the special case where the honest angle $\theta$ is either equal to 0 or $\pi/2$. Following the derivations of Eq.~(\ref{cheat_prob}) and Eq.~(\ref{guessoutput}) we have
\begin{align}\label{Prob_theta}
\text{Pr[guess }Y_H|\theta] &=\frac{1}{2}+\frac{1}{2}\sqrt{\big(\norm{\ket{\Psi_\theta}}^2+\norm{\ket{\Psi_{\theta+\pi}}}^2 \big)^2-4|\braket{\Psi_\theta}{\Psi_{\theta+\pi}}|^2} \nonumber\\
&\leq \frac{1}{2}+\frac{1}{2}\sqrt{(p_\theta + q_\theta)^2 - (p_\theta- q_\theta)^2 \sin^2(\theta + \alpha)}
\end{align}
We know that $S = p_\theta+q_\theta$ is a constant, independent of $\theta$. From Eq. (\ref{fidelity'}) and the fact that $F'(\rho) = \frac{1}{2}$, we have that $2\sqrt{p_\theta q_\theta} = 1 - S$.  We can easily infer:
\begin{align}
(p_{\theta} - q_{\theta})^2=S^2-4p_\theta q_\theta=2S-1
\end{align}
Eq. (\ref{Prob_theta}) then becomes:
\begin{align}
\text{Pr[guess }Y_H|\theta] & \leq \frac{1}{2}+\frac{1}{2}\sqrt{S^2 - (2S - 1) \sin^2(\theta + \alpha)}
\end{align}
It also holds that $F'(\rho) \le S$ which implies $S \ge \frac{1}{2}$. When $S \in [1/2,1]$, we can analytically show that $P[\rho|\theta\text{-protocol}]$ is maximal for $S = 1$. This gives
\be
P[\rho|\theta\text{-protocol}] \le \frac{1}{\pi} \int_{0}^{\pi} \frac{1}{2}+\frac{1}{2}\sqrt{1 - \sin^2(\theta + \alpha)} d\theta = \frac{1}{2} + \frac{1}{\pi} \approx 0.818.
\ee

\paragraph{}Now if the parties are running the $XY$-protocol, then instead of integrating from 0 to $\pi$, we just need to add the cases where $\theta=0$ and $\theta=\pi/2$. We have:
\begin{eqnarray}
P[\rho|XY\text{-protocol}] &= &\frac{1}{2}\big[\text{Pr[guess }Y_H|0]+\text{Pr[guess }Y_H|\frac{\pi}{2}]\big]\\
                                           & \leq & \frac{1}{2}+\frac{1}{4}\big[\cos^2(\alpha)-\sin^2(\alpha)    \big]
\end{eqnarray}
Since $\alpha$ is a characteristic of the state, and can therefore be chosen by the source, the above probability is maximized for $\alpha=-\pi/4$, and is equal to $\cos^2(\pi/8)\approx 0.854$.\\

\end{proof}


\subsection*{Loss}
\paragraph{}If the Verifier is willing to accept an individual loss rate $\lambda$, then a cheating party can profit from declaring `loss' in order to increase the probability of passing the test. We are interested to see how the two protocols behave in the presence of loss. We concentrate on checking for genuine multipartite entanglement. Since the dishonest parties have full control of the source, and in particular their part (including purification), we treat the dishonest parties as a single system. That is, we say that a source state $|\psi\rangle_{H,D}$ is genuinely multipartite entangled if it is entangled across all bipartite cuts where $D$ is treated as a single party (that is all $D$ systems are on one side of the bipartition). 

\paragraph{}Looking only at GME in this way greatly simplifies the analysis. We now wish to bound the probability of passing the test for states which are not GME, that is, there exists a partition such that $|\psi\rangle_{H,D}$ is separable. To bound this take all the honest players which are on $D$'s side of this partition, and imagine the dishonest party has control of them too, {\it i.e.} we have a bigger $D$ including these (this cannot but help the dishonest party pass the test). We thus concentrate on product states of the form $\ket{H}_H\otimes\ket{D}_D$. \\

{\bf The $\theta$-protocol.}
Let $\ket{H} = \alpha\ket{0}^{\otimes k} + e^{i\theta'} \beta \ket{1}^{\otimes k}  + \gamma \ket{\mathcal{X}}$ the state shared by the honest players with $\ket{\mathcal{X}}$ orthogonal to both $\ket{0}^{\otimes k}$ and $\ket{1}^{\otimes k}$ and $\alpha,\beta \in \mathbb{R^+}$. For a fixed $\theta$ and using the characterization of our test, the honest players will output $Y_H = 0$ with probability 
\be
\Pr[Y_H=0 | \theta ]= \frac{|\gamma|^2}{2} + |\frac{\alpha}{\sqrt{2}} + \frac{\beta}{\sqrt{2}}e^{i (\theta' - \theta)}|^2
= \frac{1}{2} + \alpha\beta \cos(\theta' - \theta)
\ee

The dishonest parties want to guess $Y_H$. They will guess $Y_H = 0$ when $\cos(\theta' - \theta) \ge 0$ and  $Y_H = 1$ otherwise, and they will succeed with probability $\frac{1}{2} + \alpha\beta|\cos(\theta' - \theta)|$. This probability is maximized for $\alpha,\beta = \frac{1}{\sqrt{2}}$. Without any loss, the dishonest players succeed with probability:
\be
\frac{1}{\pi} \left(\int_{0}^{\pi} \frac{1}{2} + \frac{1}{2}|\cos(\theta' - \theta)| d\theta\right) = \frac{2}{\pi} \int_{\theta'}^{\theta' + \pi/2} \cos^2(\frac{\theta}{2})  d\theta
\ee
In the case where there is loss, the cheating players can post-select on a $\lambda$ fraction of the angles. This is the only thing they can do since their state $\ket{D}$ is unentangled with $\ket{H}$. The worst angles are the ones close to $\pi/2 + \theta'$. In that case, when the state is tested, the cheating players pass the test with probability:
\be
P(\lambda) = \frac{2}{\pi(1-\lambda)} \int_{\theta'}^{\theta' + \pi(1-\lambda)/2} \cos^2(\frac{\theta}{2})  d\theta = 
\frac{2}{\pi(1-\lambda)} \int_{0}^{\pi(1-\lambda)/2} \cos^2(\frac{\theta}{2})  d\theta
\ee

{\bf The XY-protocol.}
Analyzing this protocol is done in a similar way as before. We start from $\ket{H} = \alpha\ket{0}^{\otimes k} + e^{i\theta'} \beta \ket{1}^{\otimes k}  + \gamma \ket{\mathcal{X}}$ with $\alpha,\beta \in \mathbb{R^+}$.
\begin{itemize}
	\item The honest parties receive an even number of Pauli $Y$ measurement requests: this corresponds to them performing a $\theta$-test with $\theta = 0$. This means that the honest players output $Y_H = 0$ with probability $\frac{1}{2} + \alpha\beta \cos(\theta')$. The optimal dishonest strategy is to guess $Y_H = 0$ when $\cos(\theta') \ge 0$. Otherwise, they guess $Y_H = 1$. This overall strategy will succeed with probability $\frac{1}{2} + \alpha\beta |\cos(\theta')|$. Notice that this is maximized for $\alpha,\beta = \frac{1}{\sqrt{2}}$ which gives $\Pr[\text{pass test} | \text{even $Y$}]= \frac{1}{2} + \frac{|\cos(\theta')|}{2}$.
	\item The honest parties receive an odd number of Pauli $Y$ measurement requests: this corresponds to them performing a $\theta$-test with $\theta = \pi/2$. Similarly as above, we can show that $\Pr[\text{pass test} | \text{odd $Y$}]= \frac{1}{2} + \frac{|\cos(\theta' + \pi/2)|}{2} = \frac{1}{2} + \frac{|\sin(\theta')|}{2} $.
\end{itemize} 

In the case when there is no loss, the pass probability of state $\ket{H}$ is maximised for $\theta'=\pi/4$, since for both measurement settings of the honest parties, the pass probability is $cos^2(\pi/8)\approx 0.854$. For $50\%$ loss, the pass probability of state $\ket{H}$ is maximised for $\theta'=0$, since whenever the dishonest party is asked to measure in the Pauli $Y$ basis, he declares loss, resulting in a pass probability equal to 1. For any amount of loss between these two values, the optimal dishonest strategy is a probabilistic mixture of the two pure strategies:
\begin{itemize}
\item With probability $2\lambda$ the source sets $\theta'=0$, and whenever the dishonest party receives $Y$, he declares loss.
\item With probability $1-2\lambda$ the source sets $\theta'=\pi/4$.
\end{itemize}
Let $Q(\lambda)$ be the probability that the dishonest parties pass the test, conditioned on not declaring loss. We have:
\begin{align}
Q(\lambda) & = \frac{1}{1-\lambda}\left(\lambda\cdot \Pr[\text{pass test} |\theta'=0,X]+(1-2\lambda) \cdot \Pr[\text{pass test} | \theta'=\pi/4]\right) \\
& =  \frac{\lambda+(1-2\lambda) \cdot 0.854}{1-\lambda}
\end{align}

\paragraph{}Supplementary Figure~\ref{fig:GHZ4vsGHZ3LossDependence} shows the difference in the pass probability for the two tests when the amount of tolerated loss increases. Here, $Q(\lambda)$ is plotted for the $XY$-protocol test and $P(\lambda)$ is plotted for the $\theta$-protocol test.

\section*{Supplementary Note 2}

\subsection*{State generation}

\paragraph{}The generation of photon pairs in our setup is achieved by spontaneous four-wave mixing (SFWM) in fiber sources exploiting birefringent phase-matching~\cite{Halder09S,Clark11S}. The fibers are strongly birefringent ($\Delta n =4 \times 10^{-4}$) and microstructured, with the phase-matching generating signal-idler pairs cross-polarization to the pump laser. The waveguide contributions to the dispersion in addition to the birefringence tailor the SFWM to the generation of naturally narrowband spectrally uncorrelated photons when pumped with Ti-Sapphire laser pulses at ~726nm. This is achieved at the flat region of the phase-matching curves upon which the idler photons ($\lambda_i = 871$~nm) are group velocity matched to the pump pulse so that they become spectrally broad ($\Delta \lambda_i =2.2$~nm) whilst the signal photons ($\lambda_s = 623$~nm) are intrinsically narrowband ($\Delta \lambda_s =0.3$~nm). This narrowband phase-matching results in a Joint-Spectral Amplitude (JSA) which is highly separable for a wide range of pump bandwidths and thus single photons of high purity can be produced. The pump bandwidth can then be tuned to minimize the effects of deviating from the flat region at $726$~nm whilst reducing the self-phase modulation caused by short pulses, to arrive at an optimal pump bandwidth of $\Delta \lambda_p = 1.7$~nm.

\paragraph{}The fiber sources are then positioned in Sagnac-loop configurations in which the pump pulse is set to diagonal polarization and split at a polarizing beam-splitter (PBS), after which it is launched into the fiber in both directions simultaneously. The $90^{\circ}$ rotation of the fiber axis between its two facets results in the pump light being strongly suppressed out of the port of the PBS it entered, whilst the generated signal-idler pairs from each facet of the fiber are cross-polarized to the pump that generated them, so exit from the other port to the pump. The pairs generated from each direction traverse the same mode in reverse so on exiting the PBS they coherently share the same spatio-temporal mode and create the state $\rhalf (\ket{H}_s\ket{H}_i + \text{e}^{i \theta} \ket{V}_s \ket{V}_i)$ up to some phase $\theta$.

\paragraph{}The generation of three- and four-photon GHZ states in our setup is then achieved by a parity check, or `fusion', with post-selection~\cite{Pittman,Pan,Pan2,Bell12S}. Fusion processes of this sort require photons originating from two distinct sources to be indistinguishable in all degrees of freedom, however the fabrication of microstructured fibres can result in small inhomogeneities between fiber samples. To overcome these inhomogeneities one fiber source is temperature tuned so that the spectra of the signal photons match the spectra of the signal photons in the other fiber. This reduces the distinguishability. The spectra of the signal photons from each source were measured for a range of pump bandwidths and Gaussians fitted to determine their central wavelength. The second source was then temperature tuned using a Peltier cooler to $23.7^{\circ}C$ (relative to the ambient $17.6^{\circ}C$) to achieve optimal indistinguishability in the spectra (see Supplementary Figure~\ref{fig:SpectraPlots}). It is useful to note here that despite the inhomogeneous distribution of heat to the fiber, which results in significant broadening of the idler photon, the narrowband phase-matching scheme ensures that the signal photon remains narrowband. However, note there are still small differences between the signal spectra that arise from inhomogeneities in the fibre and these reduce the maximum fidelity achievable.

\paragraph{}The generation of the four-photon GHZ state proceeds by overlapping the signal photons from two Bell pair sources at a PBS and post-selecting the event in which one photon is detected at each output port. On the other hand, the three-photon GHZ state requires one of the sources to contribute just a single heralded signal photon in the state $\ket{D}=\rhalf(\ket{H}+\ket{V})$ and post-selecting similarly. This is achieved by pumping the second source in only one direction and rotating the heralded signal photon with a half-wave plate.

\paragraph{}Arbitrary local projective measurements are achieved by polarisation rotations using pairs of half- and quarter-wave plates, followed by polarising beam splitters (PBSs) to spatially separate the two eigenstates of polarisation, before collection into 8 silicon avalanche photodiode detectors. Pairs of automated achromatic half- and quarter-wave plates were calibrated to account for the chromatic deviations at signal and idler wavelengths, and numerical methods were used to find wave plate angles to map the input states to the states closest to the ideal projection vectors. Note that due to the chromatic deviations of wave plates, not all rotations can necessarily be achieved, so to allow the Pauli bases and the $X$-$Y$ equator to be reached, appropriate approximate states were chosen by fiber polarizers for input to the measurement stage.

\subsection*{Higher-order terms from sources}

\subsubsection*{4-qubit GHZ}

\paragraph{}The state generated by four-wave mixing in one source in an `entangled configuration' can be written as~\cite{Fulconis07}
\bqa
\ket{\psi}_{s,i}&=&{\cal N} (\ket{0,0}_{s,i}+\alpha( \ket{1_H,1_H}_{s,i}+\ket{1_V,1_V}_{s,i}) \nonumber \\
&&\hskip1cm +\alpha^2(\ket{2_H,2_H}_{s,i}+\ket{2_V,2_V}_{s,i} +\ket{1_H1_V,1_H1_V}_{s,i})+{\cal O}(\alpha^3)),
\eqa
where ${\cal N}$ is a normalisation constant, $|\alpha|^2=\bar{n}/(\bar{n}+1)$ is the mean number of signal-idler pairs generated in a pulse and $\ket{\ell_{H/V}}_k=\frac{1}{\sqrt{\ell !}}(\hat{a}^\dag_{H/V,k})^\ell\ket{0}_k$ for mode $k$. Taking two sources in the entangled configuration we have the starting state
\bqa
\ket{\psi}_{s_1,i_1,s_2,i_2}&=&{\cal N} (\ket{0,0}_{s_1,i_1}+\alpha( \ket{1_H,1_H}_{s_1,i_1}+\ket{1_V,1_V}_{s_1,i_1}) \nonumber \\
&&\hskip0cm +\alpha^2(\ket{2_H,2_H}_{s_1,i_1}+\ket{2_V,2_V}_{s_1,i_1} +\ket{1_H1_V,1_H1_V}_{s_1,i_1}) + {\cal O}(\alpha^3)) \otimes \nonumber \\
&&(\ket{0,0}_{s_2,i_2}+\alpha( \ket{1_H,1_H}_{s_2,i_2}+\ket{1_V,1_V}_{s_2,i_2}) \nonumber \\
&&\hskip0cm +\alpha^2(\ket{2_H,2_H}_{s_2,i_2}+\ket{2_V,2_V}_{s_2,i_2} +\ket{1_H1_V,1_H1_V}_{s_2,i_2})+{\cal O}(\alpha^3)),
\eqa
which gives 35 terms when expanded up to $\alpha^3$. Applying the PBS transformations for the fusion: $\hat{a}_{H,s_1} \to \hat{a}_{H,s_1}~,\hat{a}_{V,s_1} \to \hat{a}_{V,s_2},~\hat{a}_{H,s_2} \to \hat{a}_{H,s_2}$ and $\hat{a}_{V,s_2} \to \hat{a}_{V,s_1}$, and taking terms that have at least one photon in each mode we have the state
\bqa
\hskip-1cm&&\hskip-0.6cm\ket{\psi}= {\cal N}(\alpha^2(\ket{1_H,1_H,1_H,1_H}+\ket{1_V,1_V,1_V,1_V}) +\alpha^3(\ket{2_H,2_H,1_H,1_H} \\
&&+\ket{1_H,1_H,2_H,2_H}+\ket{2_V,1_V,1_V,2_V}+\ket{1_V,2_V,2_V,1_V}+\frac{1}{2}(\ket{1_H1_V,1_H,1_V,1_H1_V} \nonumber \\
&& +\ket{1_H1_V,1_H1_V,1_V,1_V}+\ket{1_H,1_H1_V,1_H1_V,1_H}+\ket{1_V,1_V,1_H1_V,1_H1_V})))_{s_1,i_1,s_2,i_2}, \nonumber 
\eqa
where the terms with $\alpha^2$ lead to the desired GHZ state and higher-order terms with $\alpha^3$ cause the state to be non-ideal. Here we have not included the possibility of further postselection depending on the measurement basis. For example, in the H/V basis the last 4 terms can be dropped, as two photons in a single mode will lead to both detectors from the polarisation analysis of that mode giving a click. This is not the case for all bases however. 

\paragraph{}The fidelity of $\ket{\psi}$ with respect to the ideal GHZ state is $F=2 \alpha^4/(2\alpha^4+5\alpha^6)$. For the pump power used in our experiment of $P= 7$~mW in each fibre in each direction we have $\bar{n}=0.05$ and therefore $\alpha=0.22$, leading to a fidelity of $F=0.89$. Thus, higher-order emissions up to $\alpha^3$ reduce the quality of the state, as measured using the fidelilty, by $11\%$. Terms with $\alpha^4$ are $0.22$ times smaller than those with $\alpha^3$ and will therefore have a contribution of only $1-2 \%$. At the pump power used we have a rate of four-folds of 1-2 $s^{-1}$. An interesting question is whether the higher-order emissions can be used by the dishonest parties to gain an advantage when loss is present. This is a system dependent issue which we leave for future work. However, we note that regardless of this, by using a smaller pump power one can reduce the impact of higher order terms on the fidelity in our setup, although at the expense of the overall four-fold rate. For example, with $P=1$~mW one can reduce the impact on the fidelity to only $2\%$.

\subsubsection*{3-qubit GHZ}

\paragraph{} The state generated by four-wave mixing in one source in a `product configuration' can be written as~\cite{Fulconis07}
\bqa
\ket{\psi}_{s,i}&=&{\cal N} (\ket{0,0}_{s,i}+\alpha \ket{1_H,1_H}_{s,i} +\alpha^2\ket{2_H,2_H}_{s,i}+{\cal O}(\alpha^3)).
\eqa
Taking one source in the product configuration and the other in the entangled configuration we have the starting state
\bqa
\ket{\psi}_{s_1,i_1,s_2,i_2}&=&{\cal N} (\ket{0,0}_{s_1,i_1}+\alpha \ket{1_H,1_H}_{s_1,i_1} +\alpha^2\ket{2_H,2_H}_{s_1,i_1}+{\cal O}(\alpha^3)) \otimes \\
&&(\ket{0,0}_{s_2,i_2}+\alpha( \ket{1_H,1_H}_{s_2,i_2}+\ket{1_V,1_V}_{s_2,i_2}) \nonumber \\
&&\hskip0cm +\alpha^2(\ket{2_H,2_H}_{s_2,i_2}+\ket{2_V,2_V}_{s_2,i_2} +\ket{1_H1_V,1_H1_V}_{s_2,i_2})+{\cal O}(\alpha^3)). \nonumber
\eqa
which gives 20 terms when expanded up to $\alpha^3$. Applying the HWP on mode $s_1$: $\hat{a}_{H,s_1} \to \frac{1}{\sqrt{2}}(\hat{a}_{H,s_1}+\hat{a}_{V,s_1})$, and the PBS transformations for the fusion: $\hat{a}_{H,s_1} \to \hat{a}_{H,s_1}~,\hat{a}_{V,s_1} \to \hat{a}_{V,s_2},~\hat{a}_{H,s_2} \to \hat{a}_{H,s_2}$ and $\hat{a}_{V,s_2} \to \hat{a}_{V,s_1}$, and taking terms that have at least one photon in each mode we have the state (conditioned on a detection of one or more photons in mode $i_1$) 
\bqa
\hskip-1cm&&\hskip-0.6cm\ket{\psi}= {\cal N} [\frac{\alpha^2}{\sqrt{2}}(\ket{1_H,1_H,1_H}+\ket{1_V,1_V,1_V}) +\frac{\alpha^3}{\sqrt{2}}\big(\ket{1_H,2_H,2_H}+\ket{2_V,1_V,2_V} \\
&&\hskip-0.4cm +\frac{1}{2}\ket{1_H1_V,1_H,1_H1_V}+\frac{1}{2}\ket{1_V,1_H1_V,1_H1_V} +\frac{1}{\sqrt{2}} \ket{2_H,1_H,1_H}+\frac{1}{\sqrt{2}}\ket{1_V,2_V,1_V} \nonumber \\
&&\hskip-0.4cm+\ket{1_H1_V,1_V,1_V}+\ket{1_H,1_H1_V,1_H} \big) ]_{s_1,s_2,i_2}, \nonumber
\eqa
where the terms with $\alpha^2$ lead to the desired GHZ state and higher-order terms with $\alpha^3$ cause the state to be non-ideal. Here we have again not included the possibility of further postselection depending on the measurement basis.

\paragraph{}The fidelity of $\ket{\psi}$ with respect to the ideal GHZ state is $F=\alpha^4/(\alpha^4+\frac{11}{4}\alpha^6)$. For the pump power used in our experiment of $P= 7$~mW in each fibre in each direction (with the source in the product configuration only pumped in one direction) we have $\bar{n}=0.05$ and therfore $\alpha=0.22$, leading to a fidelity of $F=0.88$. Thus, higher-order emissions up to $\alpha^3$ reduce the quality of the state, as measured using the fidelity, by $12\%$. Terms with $\alpha^4$ are $0.22$ times smaller than those with $\alpha^3$ and will therefore have a contribution of only $1-2 \%$. Again, for a low pump power of $P=1$~mW one can reduce the impact of the higher order terms on the fidelity to $2\%$.

\newpage

\end{widetext}

\end{document}